\documentclass[11pt, a4paper]{article}

\usepackage[UKenglish]{babel}
\usepackage[UKenglish]{isodate}

\usepackage{todonotes}
\usepackage{a4wide}
\usepackage{amsmath, amsfonts, amsthm, amssymb, amsbsy}
\usepackage{hyperref}
\usepackage{soul}
\hypersetup{colorlinks, urlcolor=blue, citecolor=blue, linkcolor=blue}
\usepackage{tikz}
\usepackage{tkz-graph}
\usepackage[ruled, vlined]{algorithm2e}
\usepackage{enumitem}
\usepackage{caption, subcaption}
\usepackage{enumitem}
\usepackage{thm-restate}

\newtheorem{theorem}{Theorem}
\newtheorem{remark}[theorem]{Remark}
\newtheorem{lemma}[theorem]{Lemma}
\newtheorem{example}[theorem]{Example}
\newtheorem{corollary}[theorem]{Corollary}

\theoremstyle{definition}
\newtheorem{definition}[theorem]{Definition}

\def\E{\mathbb{E}}
\def\calC{\mathcal{C}}
\def\calA{\mathcal{A}}

\def\calP{\mathcal{P}}
\def\epsilon{\varepsilon}
\def\calG{\mathcal{G}}
\def\calF{\mathcal{F}}
\def\calD{\mathcal{D}}
\def\calE{\mathcal{E}}
\def\emm{\mathrm{e}}

\newcommand{\floor}[1]{\left\lfloor #1 \right\rfloor}
\newcommand{\ceil}[1]{\left\lceil #1 \right\rceil}
\newcommand{\ored}[0]{\Omega^{r}_{G, q}}
\newcommand{\zred}[0]{Z_{G, q, \beta}^{r}}

\newcommand{\tmix}[0]{\tau_{\text{mix}}}
\newcommand*\bell{\ensuremath{\boldsymbol\ell}}
\newcommand*\cm{\mathsf{CM}}

\newlist{iindent}{description}{2}
\setlist[iindent]{labelwidth=0.7cm}

 \title{Fast mixing via polymers for random graphs with unbounded degree}

\author{Andreas Galanis
\thanks{A preliminary short version of this paper (without proofs) will appear in the proceedings of RANDOM 2021. Authors' address: Department of Computer Science, University of Oxford, Wolfson Building, Parks Road, Oxford, OX1~3QD, UK.}
\and
Leslie Ann Goldberg$^*$
\and
James Stewart$^*$}

\date{9 March 2022}

\begin{document}

\maketitle{}
\begin{abstract}
The polymer model framework is a classical tool from statistical mechanics that has recently been used
to obtain approximation algorithms for spin systems on classes of bounded-degree graphs; examples include the ferromagnetic Potts model on expanders and on the grid. One of the key ingredients in the analysis of polymer models is controlling the growth rate of the number of polymers, which has been typically achieved so far by invoking the bounded-degree assumption. Nevertheless, this assumption is often restrictive and obstructs the applicability of the method to more general graphs. For example, sparse random graphs typically have bounded average degree and good expansion properties, but they include vertices with unbounded degree, and therefore are excluded from the current polymer-model framework.

We develop a less restrictive framework for polymer models that relaxes the standard bounded-degree assumption, by reworking the relevant polymer models from the edge perspective. The edge perspective allows us to bound the growth rate of the number of polymers in terms of the total degree of polymers, which in turn can be related more easily to the expansion properties of the underlying graph. To apply  our methods, we consider  random graphs with unbounded degrees from a fixed degree sequence (with minimum degree at least~$3$) and obtain approximation algorithms for the ferromagnetic Potts model, which is a standard benchmark for polymer models. Our techniques also extend to more general spin systems. 
\end{abstract}

\SetGraphUnit{1}
\GraphInit[vstyle = Simple]
\tikzset{Vtx/.style = {shape = circle, fill = black, minimum size = 6pt, inner sep = 0pt}}

\section{Introduction}

The polymer model framework~\cite{kotecky1986cluster, gruber1971general} is a classical tool from statistical mechanics which has recently been used to obtain efficient approximation algorithms for analysing spin systems (such as the Potts model) in  parameter regimes   where standard algorithmic approaches are provably inefficient/inaccurate on general graphs. These algorithms  apply to certain classes of graphs 
that typically have sufficiently strong expansion properties 
relative to their local growth rates. Typically, the local growth rate is restricted by a bounded-degree assumption.  Examples of such classes include 
bounded-degree expanders~\cite{jenssen2020algorithms, liao2019counting, chen2019fast, cannon2020counting, carlson2020efficient, barvinok2019weighted, galanis_et_al, helmuth2020finite} and
the  $d$-dimensional grid~\cite{helmuth2020algorithmic, borgs2020efficient,jenssen2020algorithms,flows}. The purpose of this work is to expand the current framework for applying polymer models by relaxing the bounded-degree assumption and using alternative methods to capture the growth of the graph.

To briefly review the current framework, we use as a running example the $q$-state ferromagnetic Potts model with parameter $\beta>0$. For a graph $G=(V_G,E_G)$, the set $\Omega_{G,q}$ of configurations of the model is the set of all (not necessarily proper)  $q$-colourings $\sigma$ of~$V_G$ using the set of colours $[q]=\{1,\hdots,q\}$ where $q\geq 3$. The weight of a colouring $\sigma$ is $w(\sigma)=\emm^{\beta m_G(\sigma)}$ where $m_G(\sigma)$ is the number of monochromatic edges under $\sigma$. The so-called partition function $Z=Z_{G,q,\beta}$ is the aggregate weight of all $\sigma$ and the Gibbs distribution $\mu=\mu_{G,q,\beta}$ is the probability distribution on the set of all $\sigma$, in which each $\sigma$ has mass proportional to its weight, i.e., $\mu(\sigma)=w(\sigma)/Z$. We will study the computational problems of approximating the partition function and approximately sampling from the Gibbs distribution.  In general, these problems are computationally hard (\#BIS-hard) when the parameter $\beta$ is sufficiently large \cite{goldberg2012,galanis2016}.

The recent works~\cite{helmuth2020algorithmic,jenssen2020algorithms} introduced a framework based on polymer models that bypasses the  worst-case hardness, on classes of bounded-degree graphs with expansion properties. The rough intuition for the Potts model is that, for large $\beta$, due to the expansion properties, the colourings with non-negligible weight are close to the so-called ground-states of the model, i.e., the $q$ configurations in which all vertices get the same colour. Polymer models capture the deviation of configurations from these ground states. Given a ground state with colour $r$, a polymer is a connected set of vertices, none of which is coloured with $r$, and a polymer configuration (with respect to the ground state~$r$) corresponds to the set of all polymers (see Example~\ref{ex:Potts} for more details). The Potts model can then be decomposed into $q$ polymer models, each of which can be studied using relatively streamlined algorithmic methods (based on  interpolation \cite{barvinok2016combinatorics} and Markov chains). This framework has already found multiple algorithmic applications in far more general settings  \cite{helmuth2020algorithmic, borgs2020efficient,flows,jenssen2020algorithms, jenssen2020homomorphisms, liao2019counting, chen2019fast, friedrich2021spectral, cannon2020counting, carlson2020efficient, galanis_et_al, helmuth2020finite}.

Despite these advances, the current applications of polymer models rely crucially on the fact that the maximum degree of the underlying graph is bounded. This fact is used   to control the number of polymers of a given size (which is crucially needed for the algorithmic analysis). As a result of this limitation, applications to several other interesting classes of graphs are ruled out, excluding for example sparse random graphs, which  have bounded average degree and good expansion properties, but include vertices with unbounded degree.

\subsection{Main Results}

In this paper, we propose a framework for polymer models that overcomes the bounded-degree limitations of previous algorithms, by revisiting the Markov chain approach of~\cite{chen2019fast}. We introduce a new condition which requires that the weight of each polymer   decays exponentially in its total degree (the sum of the
degrees of the vertices in the polymer)
instead of decaying exponentially in the polymer's size. 
This new condition allows us to prove rapid mixing for a Markov chain which is an adapted edge-version of the so-called polymer dynamics of~\cite{chen2019fast}. Crucially, the fact that the new condition is formulated in terms of the total degree of a polymer allows us to relax the assumption that the instance has
bounded degree.

As an application of our method, we consider the $q$-state ferromagnetic Potts model on sparse random graphs of unbounded degree with a given degree sequence, as detailed below.
\begin{definition}
\label{def:degseq}
Let $d$ be a positive real number and $n$ be a positive integer. We define $\calD_{n, d}$ to be the set of all degree sequences $\{x_1, x_2, ..., x_n \}$ that satisfy
\begin{enumerate}
\item \label{item:i} 
For all $i\in [n]$,
$3 \leq x_i \leq n^{\rho}$ where $\rho=\tfrac{1}{50}$, and
\item \label{item:ii} $ \sum_{i \in [n]} x_i^2 \leq d n$.
\end{enumerate} 
We write $G \sim \calG(n, \vec{x})$ to indicate that $G$ is a graph chosen uniformly at random from the set of all simple $n$-vertex graphs with  degree sequence $\vec{x}$. We say that $G$ satisfies a property with high probability (whp) if the probability that $G$ satisfies the property is $1-o(1)$, as a function of~$n$, uniformly over~$\vec{x}$. 
\end{definition}
Note that $\calD_{n,d}$ is empty unless $d\geq 9$. The assumption that all degrees are greater than or equal to 3 (rather than 2) guarantees that the random graph $G$ is connected and has good expansion properties. The degree lower bound also means that our results do not apply to Erd\H{o}s-R\'{e}nyi random graphs. The upper bound on the degrees is mild and can in fact be relaxed somewhat further (but in general cannot be made to be linear in $n$ due to the sparsity assumption in Item \ref{item:ii}).

We give an efficient algorithm for approximately sampling\footnote{
A polynomial-time approximate sampling algorithm for $\mu_{G, q, \beta}$ is an algorithm that, given an accuracy parameter $\epsilon > 0$ and a graph $G = (V_G, E_G)$ as input, outputs a sample from a probability distribution that is within total variation distance $\epsilon$ of $\mu_{G, q, \beta}$, in time poly$(|V_G|,1/\epsilon)$.} from
and approximating the partition function\footnote{
Given an accuracy parameter $\epsilon > 0$, we say that $\hat{Z}$ is an $\epsilon$-approximation to the quantity $Z$ if $\emm^{- \epsilon} Z \leq \hat{Z} \leq \emm^\epsilon Z$.
A fully polynomial randomised approximation scheme (FPRAS) for $Z_{G, q, \beta}$ is a randomised algorithm that, given an accuracy parameter $\epsilon > 0$ and a graph $G = (V_G, E_G)$ as input, outputs a random variable that is an $\epsilon$-approximation to $Z_{G, q, \beta}$ with probability at least $3/4$, in time poly$(|V_G|,1/\epsilon)$.} of the  ferromagnetic Potts model on random graphs with a given degree sequence for all sufficiently large $\beta$.

\begin{restatable}{theorem}{randgrsampling}\label{thm:randgrsampling} \label{thm:main}
Let $d$ be a real number and $q\geq 3$ be an integer. 
For the ferromagnetic Potts model, there is $\beta_0$ such that for all $\beta\geq \beta_0$ 
there is a poly-time approximate sampling algorithm for
$\mu_{G,q,\beta}$ and an FPRAS for $Z_{G,q,\beta}$ that work with high probability on random graphs $G\sim \calG(n, \vec{x})$ for
any degree sequence $\vec{x}\in \calD_{n,d}$.
\end{restatable}

\begin{remark}\label{rem:ourrem}
 Note that $\beta_0$ depends on $d$ and $q$, and our arguments later 
 (see Remark~\ref{rem:remarkdetails}) show that  $\beta_0=C d\log d\log q $ 
 for some $C>0$ (independent of~$d$ or~$q$).
If the desired accuracy $\epsilon$ is at least $ \emm^{-n}$ 
then the running time of the sampling algorithm is $O\big(n \log \tfrac{n}{\epsilon} \log\tfrac{1}{\epsilon}\big)$ and the running time of the FPRAS is 
$O\big(n^2 (\log \tfrac{n}{\epsilon})^3 \big)$.
\end{remark}

We further remark that the bounded-degree assumption has also been relaxed in \cite{flows} for the ferromagnetic Potts model on lattice graphs; the approach therein however is tailored to a certain flow representation of the Potts model, which is used as a basis for the corresponding polymer models and therefore does not extend to general spin systems. Our approach applies to general polymer models as detailed in the next section and our focus on the ferromagnetic Potts model is mainly to illustrate the method without further technical overhead; the approach for example can be adapted to general spin systems on bipartite random graphs with a given degree sequence (analogously to \cite{galanis_et_al}).

\section{Polymers}
\label{sec:conditions}

The main technique that we use to prove Theorem~\ref{thm:main} is to refine the polymer 
framework by introducing a new polymer sampling condition
which requires that the weight of each polymer
decays exponentially in its total degree.   
In order to state the new condition we first  give some relevant definitions. 

Let $G=(V_G,E_G)$ be a graph --- we refer to~$G$ as the ``host graph'' of the polymer model. Let $[q]=\{1,\hdots,q\}$ be a set of spins and $g=\{g_v\}_{v\in V_G}$ specify a set of ground-state spins for the vertices, i.e., $g_v\subseteq [q]$ for each $v\in V_G$. A polymer is a pair $\gamma = (V_\gamma, \sigma_\gamma)$ consisting of a connected set of vertices $V_\gamma$ and an assignment $\sigma_\gamma \colon V_\gamma \to [q]$ such that $\sigma_\gamma(v)\in [q]\backslash g_v$.
Let $\calP_G$ be the set of all polymers. A polymer model for the host graph $G$ is specified by a subset of allowed polymers $\calC_G \subseteq \calP_G$, and a weight function $w_G : \calC_G \rightarrow \mathbb{R}_{\geq 0}$. For polymers $\gamma, \gamma' \in \calP_G$, we write $\gamma \sim \gamma'$ to denote that  $\gamma, \gamma'$ are compatible, i.e., that for every $u\in \gamma$ and $u'\in \gamma$ the graph distance in $G$ between $u$ and $u'$ is at least 2. We define $\Omega_G = \{ \Gamma \subseteq \calC_G \mid \forall \gamma, \gamma' \in \Gamma, \gamma \sim \gamma' \}$ to be the set of all sets of mutually compatible polymers of $\calC_G$; elements of $\Omega_G$ are called polymer configurations. We define the partition function as $Z_{G} = \sum_{\Gamma \in \Omega_G} \prod_{\gamma \in \Gamma} w_{G}(\gamma)$, and the Gibbs distribution on $\Gamma\in \Omega_G$ as $\mu_{G}(\Gamma) =  {\prod_{\gamma \in \Gamma} w_{G}(\gamma)}/{Z_{G}}$. We use $(\calC_G, w_G)$ to denote the polymer model.

\begin{example}[The polymer model $(\calC^{r}_{G, q}, w_{G, \beta})$, \cite{jenssen2020algorithms}]\label{ex:Potts}
Consider the $q$-state ferromagnetic Potts model with parameter $\beta$, and let $r\in [q]$ be a colour. Let $G$ be a graph and set $g_v=\{r\}$ for every $v\in V_G$.   The weight of a polymer $\gamma =(V_\gamma,\sigma_\gamma)$ is defined as
$w_{G, \beta}(\gamma) = \emm^{- \beta B_\gamma}$,
where $B_\gamma$ denotes the number of edges from $V_\gamma$ to $V_G\backslash V_\gamma$ plus the number of edges of $G$ with both endpoints in $V_\gamma$ that are bichromatic under $\sigma_\gamma$.  We let $\calP^{r}_{G, q}$ denote the set of all polymers and the set of allowed polymers $\calC^{r}_{G, q}$ to be the set of polymers $\gamma \in \calP^{r}_{G, q}$ such that $|V_\gamma| < |V_G|/2$. Note that a polymer configuration $\Gamma$ consisting of the polymers $\gamma_1,\hdots,\gamma_k$ corresponds to a colouring $\sigma$, where a vertex $v$ takes the colour $\sigma_{\gamma_i}(v)$ if $v\in V_{\gamma_i}$ for some $i\in [k]$, and the colour $r$ otherwise; moreover, $e^{\beta |E_G|}\prod_{i\in [k]}w_G(\gamma_i)=w_G(\sigma)$.
\end{example}

Polymer models have been used to approximate the partition function of spin systems on bounded-degree host graphs.  There are several existing algorithmic frameworks which can be used to sample from these resulting polymer models. One such deterministic algorithm uses the polynomial interpolation method of Barvinok~\cite{barvinok2016combinatorics} combined with the cluster expansion to approximate the partition function of the polymer model, see~\cite{helmuth2020algorithmic} for more details. Typical running times for these deterministic algorithms are of the form $n^{O(\log(\Delta))}$, where $\Delta$ is the maximum degree of the host graph, though for polymer models these have been improved to give a running time of $n^{1+o_\Delta(1)}$, see \cite{jenssen2020algorithms}. Another approach, which we will describe in detail in Section~\ref{sec:details}, uses a Markov chain called the polymer dynamics to sample from $\mu_G$ (see also~\cite{chen2019fast} for more details). The running times of algorithms obtained using the Markov chain approach are usually faster and of the form $O(n \log n)$. Both of these approaches work for restricted ranges of parameters,\footnote{
For example, for the $q$-colour ferromagnetic Potts model
with inverse temperature $\beta$ on $\alpha$-expander graphs 
of maximum degree~$\Delta$ where $\alpha>0$, $\Delta \geq 3$, $q\geq 2$,
Theorem~3 of \cite{jenssen2020algorithms}  applies when $\beta \geq (4 + 2 \log(q \Delta))/\alpha$
whereas
Theorem 9 of \cite{chen2019fast}
 requires $\beta \geq (5 + 3 \log((q-1)\Delta)/\alpha$.}
and the essential condition required is that the weight of each polymer decays exponentially in the number of vertices it contains. To obtain our results, we give a simple generic way to 
modify this condition, as detailed below.

For a vertex $v \in V_G$ we write $\deg_G(v)$ to denote the degree of $v$ in $G$, and for a vertex subset $S \subseteq V_G$ we write $\deg_G(S)$ to denote $\sum_{v \in S} \deg_G(v)$.
\begin{definition}
\label{def:psample}
Let $q\geq 2$ be an integer, $\calG$ be a class of graphs, and
$\calF_\calG = \{ (\calC_G, w_G) \mid G \in \calG \}$ be a family of $q$-spin polymer models. We say that $\calF_\calG$ satisfies the \textit{polymer sampling condition} with constant $\tau \geq 3 \log(8 \emm^3 (q-1))$ if 
$
w_G(\gamma) \leq \emm^{- \tau \deg_G(V_\gamma)}
$
for all $G \in \calG$ and all $\gamma \in \calC_G$.\footnote{Unless we specify otherwise, the base of all logarithms in this paper is assumed to be~$\emm$.}
\end{definition}

Using Definition~\ref{def:psample},
we will  show (Lemma~\ref{lem:pmersampling}, below)
that 
if a ``computationally feasible'' family of polymer models on a class of graphs $\calG$ satisfies  this new condition, then there is an efficient algorithm which, given a graph $G \in \calG$, approximately samples from the Gibbs distribution of the polymer model corresponding to $G$.

The new polymer sampling condition in Definition~\ref{def:psample} is analogous to the original one
in~\cite{chen2019fast}
except that the original condition requires the weight of a polymer to decay exponentially in its size, and in particular that the constant $\tau$ is sufficiently big relative to the maximum degree of $G$. 
The new condition relaxes this, allowing us to choose the constant~$\tau$ in Definition~\ref{def:psample} so that it
 does not depend  on the maximum degree of the host graph, which is how we overcome the limitations of previous work. Technically,
the improvement comes from the fact that
previous work relies on bounding   the number of connected vertex subsets of a given size 
(with bounds that depend on the maximum degree of the graph),  but here we are able to 
instead rely on 
the following lemma which bounds the number of connected vertex subsets with a given total degree
and this enables us to avoid  restricting
the maximum degree of the graph. 
The new condition, which replaces the notion of ``size'' with total degree,
fits well with the original abstract polymer model framework of~\cite{kotecky1986cluster}, where the notion of 
the ``size'' of a polymer is an abstract function.

\begin{restatable}{lemma}{npolymers}\label{lem:npolymers}
Let $G = (V_G, E_G)$ be a graph, $v \in V_G$, and $\ell \geq 1$ be an integer. The number of connected vertex subsets $S \subseteq V_G$ such that $v \in S$ and $\deg_G(S) = \ell$ is at most $(2\emm)^{2\ell - 1}$.
\end{restatable}

In addition to the bound on the number of connected vertex subsets in
Lemma~\ref{lem:npolymers}, we will use the fact that these connected vertex subsets can be enumerated in time exponential in the total degree~$\ell$ (see Lemma~\ref{lem:enumsgraphs}). Although the bound in Lemma~\ref{lem:npolymers} is exponential in~$\ell$, this will be mitigated by the fact that the new polymer sampling condition ensures that
the weight of each polymer is  exponentially small in its total degree. The new polymer sampling condition therefore allows us to prove that the following condition holds --- this condition is  analogous to the polymer mixing condition of~\cite{chen2019fast}, except that we consider edges instead of vertices.  For a polymer $\gamma \in \calP_G$, let $E_\gamma$ denote the set of edges of $G$ with at least one endpoint in $V_\gamma$.
\begin{definition}
\label{def:pmix}
Let $q\geq 2$ be an  integer, $\calG$ be a class of graphs, and $\calF_\calG = \{ (\calC_G, w_G) \mid G \in \calG \}$ be a family of $q$-spin polymer models. We say that $\calF_\calG$ satisfies the \textit{polymer mixing condition} with constant $\theta \in (0, 1)$ if
$
\sum_{\gamma' \nsim \gamma} |E_{\gamma'}| \cdot w_{G}(\gamma') \leq \theta |E_\gamma|
$
for all $G \in \calG$ and all $\gamma \in \calC_G$.
\end{definition}

In contrast to the conditions in~\cite{chen2019fast}, the two new conditions consider edges since we modify the polymer dynamics algorithm to sample edges instead of vertices. Subject to these new conditions, the techniques of~\cite{chen2019fast} can be adapted to show that the modified polymer dynamics mixes rapidly, therefore giving the efficient algorithm for sampling from the Gibbs distribution of a polymer model. We give the relevant details in Section~\ref{sec:details}.

Finally, in order to use the modified polymer dynamics as an efficient algorithm for computing an approximate sample from $\mu_G$, 
we will need a mild computational condition for polymers. More precisely, we say that a family of polymer models $\{ (\calC_G, w_G) \mid G \in \calG \}$ is \textit{computationally feasible} if for all $G \in \calG$ and all $\gamma \in \calP_G$,
it is possible to
decide whether $\gamma \in \calC_G$ and to compute $w_G(\gamma)$, if it is, in $O(\emm^{\deg_G(V_\gamma)})$ time. Computational feasibility serves exactly the same purpose 
as it did in Definition 3 of~\cite{chen2019fast}, which requires that the same operations are able to be carried out in time depending on $|V_\gamma|$ (instead of $\deg_G(V_\gamma)$ that we use here).

In Section~\ref{sec:details}, we 
prove the following lemma which gives an efficient algorithm for sampling\footnote{Given an accuracy parameter $\epsilon > 0$, we say that  a random variable $X$ is an $\epsilon$-sample from the probability distribution $\mu$ if the total variation distance between the distribution of $X$ and $\mu$ is at most $\epsilon$.} from the Gibbs distribution of a polymer model and for approximating its partition function.
In order to prove the lemma, we extend  the polymer dynamics algorithm of~\cite{chen2019fast} to the unbounded degree setting. The proof of the lemma uses the fact (Lemma~\ref{lem:samplingtomixing}) that the polymer sampling condition implies the polymer mixing condition. 

\begin{restatable}{lemma}{pmersampling}\label{lem:pmersampling}
Let $q\geq 2$ be an integer, $\calG$ be a class of graphs, and  $\calF_\calG$ be a family of computationally feasible $q$-spin polymer models satisfying the polymer sampling condition. 

There are randomised algorithms which, given as input a graph $G \in \calG$ with $m$ edges and an accuracy parameter $\epsilon > 0$, output an $\epsilon$-sample from $\mu_G$ in $O\big(m \log \tfrac{m}{\epsilon}\log\tfrac{1}{\epsilon}\big)$ time, and  an $\epsilon$-approximation to $Z_G$, with probability at least $3/4$, in $O\big(m^2 \log (\tfrac{m}{\epsilon})^3\big)$ time.
\end{restatable}

\section{Application to unbounded-degree graphs}

Let $\alpha > 0$ be a real number. We say that a graph $G$ is an $\alpha$-total-degree expander if, for all connected
vertex subsets $S \subseteq V_G$ with $|S| \leq |V_G|/2$, we have $e_G(S, S^c) \geq \alpha \deg_G(S)$, where  $e_G(S, S^c)$ denotes the number of edges with one endpoint in $S$ and the other in $S^c:=V_G \setminus S$.  Let $\calG_{\alpha}$ denote the set of all $\alpha$-total-degree expanders. Note, every connected $G\in \calG_\alpha$   is also an $\alpha$-expander (i.e., $e_G(S, S^c) \geq \alpha |S|$).

When $\beta$ is sufficiently large,
the polymer model from Example~\ref{ex:Potts}  satisfies the polymer sampling condition (Definition~\ref{def:psample}) with constant $\tau = \alpha \beta$. To see this, consider $\gamma \in \calC^{r}_{G, q}$ and observe that since $B_\gamma \geq e_G(V_\gamma, V_\gamma^c)$ and $|V_\gamma| < |V_G|/2$, it follows that
\begin{equation}
\label{eq:wtpolymer}
w_{G, \beta}(\gamma) \leq \exp \left\{ - \alpha \beta \deg_G(V_\gamma) \right\} = \emm^{- \tau \deg_G(V_\gamma)},
\end{equation}
where $\tau \geq 3 \log(8 \emm^3 (q-1))$ if $\beta \geq \frac{3}{\alpha} \log(8 \emm^3 (q-1))$.

We may therefore apply Lemma~\ref{lem:pmersampling} in order to efficiently sample from the ferromagnetic Potts model 
and to estimate $Z_G$ for
 $G\in \calG_{\alpha}$, provided that $\beta$ is sufficiently large. The proof of the following theorem is deferred until Section~\ref{sec:pottsdetails}.

\begin{restatable}{theorem}{expandersampling}\label{thm:expandersampling}
\sloppy Let $\alpha > 0$ be a real number. Let $q \geq 3$ be an integer and $\beta \geq \frac{3}{\alpha} \log(8 \emm^3 (q-1))$ be a real. 
For the Potts model on $G\in \calG_{\alpha}$, there is a poly-time approximate sampling algorithm for $\mu_{G,q,\beta}$ and an FPRAS for 
$Z_{G,q,\beta}$. 

In fact, for $n=|V_G|$ and $m=|E_G|$, if the desired accuracy $\epsilon$ satisfies  $\epsilon \geq \emm^{-n}$ then the running time of the sampler is  $O(m \log \tfrac{m}{\epsilon} \log \tfrac{1}{\epsilon})$ and the running time of the FPRAS is  $O\big(m^2 (\log \frac{m}{\epsilon})^3\big)$.
\end{restatable}

\subsection{Expansion of random graphs with specified degree sequences}
\label{sec:mainproof}
Let $d$ be a real number. In this section, we will show that a graph $G\sim\calG(n, \vec{x})$ for a degree sequence $\vec{x} \in \calD_{n, d}$ is whp an $\alpha$-total-degree-expander for some constant $\alpha>0$, i.e.,  that $G \in \calG_\alpha$.

To work with $G\sim \calG(n, \vec{x})$, we consider the standard configuration model, where a random multigraph  $H = (V_H, E_H)$ with the given degree sequence $\vec x$ is sampled by the following process.  For each $i \in [n]$, we attach $x_i$ half-edges to the vertex $i$. We then sample a uniformly random perfect matching on the half-edges to give $E_H$. This uniformly random perfect matching can be sampled by performing the following until no half-edges remain: choose any remaining half-edge, choose another remaining half-edge uniformly at random, then pair these two half-edges and remove them from the set of remaining half-edges. We write  $H \sim \cm(n, \vec{x})$. Note, for two vertices $i, j \in V_H$ such that $i \neq j$,  the probability that a half edge attached to $i$ and a half edge attached to $j$ are paired is
\begin{equation}\label{def:pair}
p_{\{i,j\}} = \frac{x_i x_j}{2m - 1},  \mbox{where $m = \tfrac12\sum_{k=1}^n x_k$,}
\end{equation}
and similarly the probability that two half-edges of $i$ are connected is $p_{\{i,i\}} = \tfrac{x_i(x_i - 1)}{2m(2m-1)}$.

We first prove results about $\cm(n, \vec{x})$, since asymptotic properties of $\cm(n, \vec{x})$ can easily be transferred back to $\calG(n, \vec{x})$ using the following  straightforward consequence of~\cite[Theorem 1.1]{janson2014probability}.

\begin{lemma}
\label{lem:cmtog}
Let $d$ be a positive real number. For every positive integer~$n$, let $\calE_n$ be a set of $n$-vertex multigraphs. If, for some   $\vec{x} \in \calD_{n, d}$, we have $G \sim \calG(n, \vec{x})$ and $H \sim \cm(n, \vec{x})$ then the following is true. If $H \in \calE_n$ with high probability, then $G \in \calE_n$ with high probability.
\begin{proof}
Suppose that $H'$ is drawn from $\cm(n, \vec{x})$ conditioned on being simple. It is well-known (for example, see~\cite[Proposition 7.15]{van2016random}) that $H'$ is a uniformly-random graph with degree sequence $\vec{x}$. Thus, it follows that
\begin{equation}
\label{eq:condprob}
\Pr(G \notin \calE_n) = \Pr(H' \notin \calE_n) = \frac{\Pr(H \notin \calE_n, H \text{ is simple})}{\Pr(H \text{ is simple})} \leq \frac{\Pr(H \notin \calE_n)}{\Pr(H \text{ is simple})}.
\end{equation}
By assumption, we have that $\Pr(H \notin \calE_n) = o(1)$. Applying Theorem 1.1 of~\cite{janson2014probability}, whose conditions are satisfied by Item~\ref{item:ii} of Definition~\ref{def:degseq}, it follows that there is a positive~$p$ such that when $n$ is sufficiently large $\Pr(H \text{ is simple}) > p$. By~\eqref{eq:condprob}, we therefore have that $\Pr(G \notin \calE_n) = o(1)$, and the result follows.
\end{proof}
\end{lemma}

For a (multi)graph $H = (V_H, E_H)$ we define the tree-excess to be $t_H = |E_H| - (|V_H| -1)$; that is, the number of edges more than a tree that $H$ has. First, we show that multigraphs drawn from the configuration model have locally bounded tree excess.
\begin{lemma}
\label{lem:excess}
Let $d$ be a positive real number. The following is true with high probability when $H = (V_H, E_H)$ is drawn from $\cm(n, \vec{x})$ uniformly over all degree sequences $\vec{x}   \in \calD_{n, d}$. For all connected vertex sets $S \subseteq V_H$ with $|S| \leq (\log n)^2$ and $\deg_H(S) \geq 36$, we have that $t_{H[S]} \leq \frac{1}{6} \deg_H(S)$.
\begin{proof}
For positive integers $k$ and $\ell$, and a non-negative integer $t$, let the random variable $X_{k, \ell, t}$ denote the number of connected vertex subsets $S \subseteq V_H$ such that $|S|=k$, $\deg_H(S) = \ell$, and $t_{H[S]}=t$. To prove the lemma, we will show that whp
\[
\sum_{k \leq \floor{(\log n)^2}} \sum_{\ell \geq 36} \sum_{  t \geq \lfloor \ell/6 \rfloor + 1  } X_{k, \ell, t} = 0.
\]
In fact, we can further restrict the range of summation. From the lower bound in Item~\ref{item:i} of Definition~\ref{def:degseq}, we have that $x_i \geq 3$ for all $i$, and therefore $\ell \geq 3k$. Item~\ref{item:ii} shows that $\sum_{i} x_i \leq d n$, and therefore $\ell \leq d n$ and $t \leq \ell/2 \leq d n/2$. So, consider any integer $\ell$ in the range $36 \leq \ell \leq d n$, any integer $k$ in the range $1 \leq k \leq \min \{ (\log n)^2, \ell/3 \}$, and any integer $t > \ell/6$. There are at most $\binom{n}{k}$ vertex subsets $S \subseteq V_H$ with $|S| = k$ and $\deg_G(S) = \ell$. 
Let $j=k-1+t$ be the number of edges with both endpoints in~$S$.
Given such a set~$S$, there are at most $\binom{\ell}{2j}$ possibilities for the set of half-edges in these $j$~edges. On a given set of $2j$  half-edges, there are $(2j-1)!! = \frac{(2j)!}{2^{j} j!}$ perfect matchings. Using the upper bound on the degrees from Item~\ref{item:i} of Definition~\ref{def:degseq},
the probability that a set of $j$ edges is present in~$H$ is at most 
$$
\frac{n^{2\rho}}{2m - 1} \frac{n^{2\rho}}{2m - 3} \cdots \frac{n^{2\rho}}{2m - 2j + 1} \leq \left( \frac{n^{2\rho}}{2m - 2j} \right)^{j} \leq \left( \frac{n^{2\rho}}{n} \right)^{j},
$$
where the final inequality follows from the fact that $k \leq (\log n)^2$ and therefore that $2m-2j   \geq \deg_G(S^c) \geq 3 |S^c| =3(n-k)> n$ (as long as~$n$ is sufficiently big). We also have that
$$
\binom{\ell}{2j} \cdot \frac{(2j)!}{2^{j} j!} < \frac{\ell !}{(\ell - 2j)! j!} 
< \frac{\ell^{2j} \emm^{j}}{j^{j}} \\
\leq \left( \frac{\emm \ell^2}{t} \right)^{j}.
$$
Putting everything together, it follows that
 $$
\E[X_{k, \ell, t}] \leq \binom{n}{k} \left( \frac{\emm \ell^2}{t} \right)^{k - 1 + t} \left( \frac{n^{2\rho}}{n} \right)^{k - 1 + t} \\
< \left( \frac{\emm^2 \ell^2}{t} \right)^{k - 1 + t} \frac{n^{2\rho(k - 1 + t)}}{n^{t - 1}}.
 $$ 
Furthermore, since $t > \ell/6$, $k < 2t$, and 
(by the upper bound in Item~\ref{item:i} of Definition~\ref{def:degseq})
$\ell \leq k n^{\rho}   \leq  n^{2\rho}$, 
we have that
\begin{equation}
\label{eq:expbd}
\E[X_{k, \ell, t}] < \frac{(6\emm^2  n^{4\rho})^{3t-1}}{n^{t-1}}.
\end{equation}
Let
\[
X = \sum_{\ell = 36}^{d n} \sum_{k = 1}^{\floor{\min\{ (\log n)^2, \ell/3 \}}} \sum_{t = \floor{\ell/ 6} + 1}^{\floor{d n/2}} X_{k, \ell, t}.
\]
Since $t > \ell/6 \geq 6$, it follows that $t \geq 7$.
For big enough~$n$, \eqref{eq:expbd} 
shows that 
$\E[X_{k, \ell, t}] \leq n^{13\rho t}/n^{t-1}
$.   
Since $\rho \leq 2/91$
and $t\geq 7$, $1-13\rho \geq 5/7 \geq 5/t$
so 
$13\rho t \leq t-5$ and $\E[X_{k, \ell, t}]$ is at most $n^{-4}$.  Taking a union bound over all permissible values for $\ell$, $k$, and $t$, we find that $\E[X] = o(1)$. Applying Markov's inequality, we have that $\Pr(X > 0) = \Pr(X \geq 1) \leq \E[X]= o(1)$, and the result follows.
\end{proof}
\end{lemma}

To obtain the expansion bounds in  Lemmas~\ref{lem:smallexpansion} and \ref{lem:bigexpansion}, we will use the following result from Fountoulakis and Reed~\cite{fountoulakis2007evolution}. Although this result is stated in \cite{fountoulakis2007evolution} in terms of the random graph model, it is first proved for the configuration model, so this is how we state it. Also, the result in \cite{fountoulakis2007evolution} requires that the vector $\vec{x}$ be in $\calD_{n, d}$ but this is only important for lifting their result to the random graph model, so it is
not relevant for us.

\begin{lemma}[{\cite[Proposition 4.5]{fountoulakis2007evolution}}]
\label{lem:frprob}
When   $H = (V_H, E_H)$ is drawn from $\cm(n, \vec{x})$ for some length-$n$ degree sequence~$\vec{x}$ the following is true for any set $S \subseteq V_H$.
\[
\Pr(e_H(S, S^c) = 0)  \leq	  \binom{m}{\deg_H(S)/2}^{-1},
\]
where $m = \tfrac12\sum_i x_i$. 
\end{lemma}
Note that Lemma~\ref{lem:frprob} was stated in \cite{fountoulakis2007evolution} for $S$ such that $\deg_H(S)$ is even, but if $\deg_H(S)$ is odd, it is not possible to have $e_H(S,S^c)=0$. 
Next, we show that in a multigraph~$H$ drawn from the configuration model, small vertex subsets satisfy certain expansion properties.

\begin{lemma}
\label{lem:smallexpansion}
Let $d$ be a positive real number. The following is true with high probability when $H = (V_H, E_H)$ is drawn from $\cm(n, \vec{x})$ uniformly over all degree sequences  $\vec{x}   \in \calD_{n, d}$. For all connected vertex sets $S \subseteq V_H$ with $|S|  \leq (\log n)^2$, we have that $e_H(S, S^c) \geq |S|/4$. 
\end{lemma}

\begin{proof}
For positive integers $k$ and $\ell$, and a non-negative integer $j$, let the random variable $X_{k, j, \ell}$ denote the number of connected vertex subsets $S \subseteq V_H$ with $|S|=k$, $e_H(S, S^c) = j$, and $\deg_H(S) = \ell$. By   Item~\ref{item:i} of Definition~\ref{def:degseq}, we need only consider 
$\ell$ satisfying
$3k \leq \ell \leq k   n^\rho$.  
Let
\[
X = \sum_{k =1}^{\floor{(\log n)^2}} \sum_{j =0}^{ \floor{k/4}} \sum_{\ell =3k}^{ \floor{k n^{\rho}}} X_{k, j, \ell}.
\]
To prove the lemma we will show that $X = 0$, whp. Consider any integer $k$ in the range $1 \leq k \leq (\log n)^2$, any integer $j$ in the range $0 \leq j < k/4$, and any integer $\ell$ in the range $3k \leq \ell \leq k n^\rho$. There are at most $\binom{n}{k}$ candidates for vertex sets $S$ with $|S|=k$ and $\deg_H(S) = \ell$. There are then at most $\binom{\ell}{j}$ choices for the $j$ half-edges emanating from vertices of $S$ that will be matched with half-edges emanating from vertices of $S^c$, once $H$ is drawn.  Applying  Lemma~\ref{lem:frprob} to the degree sequence derived from~$\vec{x}$ by removing the $j$ half-edges (and their partners), the probability that the remaining $\ell - j$ half-edges are matched amongst themselves is at most
\[
\binom{m'}{(\ell - j)/2}^{-1}  \leq  \left( \frac{(\ell - j)}{2m'} \right)^{\frac{(\ell - j)}{2}} \leq  \left( \frac{k n^\rho}{n} \right)^{\frac{11k}{8}} ,
\]
where $2m' = (\sum_{i=1}^n x_i) - 2j$ and the last inequality follows 
(for big enough~$n$) 
since $11k/4 \leq \ell - j \leq k n^\rho$ and $2m' \geq 3n - 2j > n$. We therefore have that
\begin{align*}
\E[X] &\leq \sum_{k = 1}^{\floor{(\log n)^2}} \sum_{j = 0}^{\floor{k/4}} \sum_{\ell = 3k}^{\floor{k n^\rho}} \binom{n}{k}  \binom{\ell}{j} \left( \frac{k n^\rho}{n} \right)^{\frac{11k}{8}} \\
&\leq \sum_{k = 1}^{\floor{(\log n)^2}} \sum_{j = 0}^{\floor{k/4}} \sum_{\ell = 3k}^{\floor{k n^\rho}} \left( \frac{n\emm}{k} \right)^k \left( \frac{\emm \ell}{j} \right)^j \left( \frac{k n^\rho}{n} \right)^{\frac{11k}{8}} \\
&\leq \sum_{k = 1}^{\floor{(\log n)^2}} \left( \frac{(\log n)^{O(1)}n^{\rho(2+11/8)}}{n^{3/8}} \right)^k.
\end{align*}
This is $o(1)$ since $\rho < 1/9$.
Applying Markov's inequality, we have that $\Pr(X > 0) = \Pr(X \geq 1) = o(1)$, and the result follows.
\end{proof}

We will also require the following lemma.

\begin{lemma}
\label{lem:bigexpansion}
Let $d$ be a positive real number. There is a positive real number $\alpha$ (depending on~$d$) such that the following is true with high probability when $H = (V_H, E_H)$ is drawn from $\cm(n, \vec{x})$ uniformly over all degree sequences  $\vec{x}   \in \calD_{n, d}$. For all connected vertex sets $S \subseteq V_H$ with $(\log n)^2 \leq |S| \leq n/2$, we have that $e_H(S, S^c) \geq \alpha \deg_H(S)$.
\begin{proof}
We split the proof into three cases based on the degree of the vertex set $S \subseteq V_H$:
\begin{enumerate}
\item $\deg_H(S) \leq 100d |S|$,
\item $\deg_H(S) > 100d|S|$ and $\deg_H(S) \leq n/2$, and
\item $\deg_H(S) > 100d|S|$ and $\deg_H(S) > n/2$.
\end{enumerate}
Throughout, we let $m=\tfrac12\sum_{i=1}^n x_i$.

\textbf{Case 1.} 
The proof of this case is similar to that of Lemma 4.1 of~\cite{fountoulakis2007evolution}, but we include it for completeness.
We will show that there is a positive real number $\alpha' < 1/4$ such that whp, every vertex set $S \subseteq V_H$ with $(\log n)^2 \leq |S| \leq n/2$ and $\deg_H(S) \leq 100d |S|$ satisfies $e_H(S, S^c) \geq \alpha' |S| \geq \alpha' \deg_H(S)/ (100d)$. We will eventually require $\alpha'$ to be sufficiently small depending on~$d$.

For positive integers $k$ and $\ell$, and a non-negative integer $j$, let the random variable $X_{k, \ell, j}$ denote the number of connected vertex subsets $S \subseteq V_H$ with $|S| = k$, $e_H(S, S^c) = j$, and $\deg_H(S) = \ell$. Let
\[
X = 
\sum_{k = \ceil{(\log n)^2}}^{\floor{n/2}}
\sum_{j=0}^{\lfloor \alpha' k \rfloor}
\sum_{\ell=1}^{100d k} X_{k, \ell, j}.
\]
Our aim will be to show that $\E[X] = o(1)$, and we begin by further restricting the range of summation in the above. Let $S \subseteq V_H$ be a connected vertex subset with $|S| = k$, $e_H(S, S^c) = j$, and $\deg_H(S) = \ell$. Since there are $(\ell - j)/2$ edges in $e_H(S, S)$, we need only consider $\ell$ such that $\ell - j$ is even. Furthermore, since $S$ is connected, we can assume that $e_H(S, S) \geq k - 1$, so we need only consider values of~$\ell$ satisfying $\ell \geq j + 2(k - 1)$. Finally, since $x_i \geq 3$ for all $i$, and $\sum_i x_i = 2m$, we have that $2m \geq \ell + 3(n - k)$, so we need only consider values of~$\ell$ satisfying $\ell \leq 2m - 3(n - k)$. Let $J_k = \{ j \in \mathbb{Z} : 0 \leq j \leq \alpha' k \}$ and let $L_{k, j} = \{ \ell \in \mathbb{Z} : j + 2(k - 1) \leq \ell \leq \min \{ 2m - 3(n - k), 100d k \},~\ell - j \text{ is even} \}$, we can therefore re-write $X$ as follows
 \[
 X = \sum_{k = \ceil{(\log n)^2}}^{\floor{n/2}} \sum_{j \in J_k} \sum_{\ell \in L_{k, j}} X_{k, \ell, j}.
 \]

Given a set of $k$ vertices, there are at most $\binom{\ell}{j}$ ways to choose the half-edges that will make up $E_H(S, S^c)$ once $H$ is drawn. The remaining total degree of $S$ after choosing these edges is $\ell - j$. By Lemma~\ref{lem:frprob}, the probability that every other edge of of $E_H(S, V_H)$ is an edge of $E_H(S, S)$ is at most $\binom{m'}{(\ell-j)/2}^{-1}$, where $m' = m - j$.  Since  $\ell \in L_{k,j}$, $(\ell-j)/2 \geq k-1$ and  $ (\ell - j)/2 \leq m - 3(n - k)/2 - j/2 $. Since $j < \alpha' k < k/4 \leq n/8$ it follows that $n - k > j$ and therefore that $  (\ell - j)/2 \leq m - 3(n - k)/2 - j/2 \leq m' - n + k$. Thus, we have that
\[
\binom{m'}{(\ell - j)/2}^{-1} \leq \max \left \{ \binom{m'}{k - 1}^{-1}, \binom{m'}{n - k}^{-1} \right \}.
\]
Using Stirling's approximation,  there exists a constant $C$ (see \cite[Equation 7]{fountoulakis2007evolution}) such that for all $n$, $m'$, and $k$ 
\[
\frac{\binom{n}{k}}{\binom{m'}{k - 1}} \leq C \cdot \frac{m'}{k} \left( \frac{n}{m'} \right)^{k}
\quad
\text{and }
\quad
\frac{\binom{n}{k}}{\binom{m'}{n - k}} \leq C \cdot \frac{m'}{m' - n + k + 1} \left( \frac{n}{m'} \right)^{n - k}.
\]
Since $2m' = 2m - 2j \geq 3n - n/4$, it follows that $n/m' \leq 8/11$. Since $k \leq n/2$, it follows that $n - k \geq k$. Thus,
\[
\frac{\binom{n}{k}}{\binom{m'}{k - 1}} \leq C \cdot m' \left( \frac{8}{11} \right)^{k}
\quad
\text{and }
\quad
\frac{\binom{n}{k}}{\binom{m'}{n - k}} \leq C \cdot m' \left( \frac{8}{11} \right)^{k},
\]
therefore 
\[
\binom{n}{k} \binom{m'}{(\ell - j)/2}^{-1} \leq C \cdot m' \left( \frac{8}{11} \right)^k.
\]
From the above, we have that
\begin{align*}
\E[X] &= O(1) \sum_{k = \ceil{(\log n)^2}}^{\floor{n/2}} \sum_{j \in J_k} \sum_{\ell \in L_{k, j}} \binom{n}{k} \binom{\ell}{j} \binom{m'}{(\ell - j)/2}^{-1} \\
&= O(1) \sum_{k = \ceil{(\log n)^2}}^{\floor{n/2}} \sum_{j \in J_k} \sum_{\ell \in L_{k, j}} m'  \binom{\ell}{j} \left( \frac{8}{11} \right)^k.
\end{align*}

We wish to show that for each term $\binom{\ell}{j} \leq (9/8)^k$. Since $j \leq \alpha' k\leq \alpha' \ell<\ell/4$ we have $\binom{\ell}{j} \leq \binom{\ell}{\alpha' k} \leq  (\tfrac{\emm \ell}{ \alpha'k})^{\alpha' k}$ and since $\ell \leq 100d k$ this is at most $(100\emm d/\alpha')^{\alpha' k}$ which is at most $(9/8)^k$ if $\alpha'$ is sufficiently small, depending on~$d$.
(It suffices to take $\alpha' = 10^{-4}/\log d$, for example.)

We conclude that we can upper bound each term by  $m' (9/11)^k \leq d n(9/11)^k$ and therefore
\[
\E[X] = O(n) \sum_{k = \ceil{(\log n)^2}}^{\floor{n/2}} \sum_{j \in J_k} \sum_{\ell \in L_{k, j}} \left( \frac{9}{11} \right)^k = O(n) \sum_{k = \ceil{(\log n)^2}}^{\floor{n/2}} k^2 \left( \frac{9}{11} \right)^k.
\]
It follows from the above that $\E[X] = o(1)$. Applying Markov's inequality, we have that $\Pr(X > 0) = \Pr(X \geq 1) = o(1)$, and the result follows.

\textbf{Case 2.} We will do more than is required, and show that every connected vertex set $S \subseteq V_H$ with $(\log n)^2 \leq \deg_H(S) \leq n/2$ satisfies $e_H(S, S^c) \geq \deg_H(S)/100$.

For every positive integer $\ell$ 
satisfying $(\log n)^2 \leq \ell \leq n/2$
and every non-negative integer $j$ satisfying $j \leq \ell/100$, 
let $X_{\ell, j}$ denote the number of connected vertex subsets $S \subseteq V_H$ with $e_H(S, S^c) = j$ and $\deg_H(S) = \ell$. Let $J_\ell = \{ j \in \mathbb{Z} : 0 \leq j \leq \ell/100, \ell - j \text{ is even} \}$ and let
\[
X = 
\sum_{\ell = \ceil{(\log n)^2}}^{\floor{n/2}}
\sum_{j \in J_\ell}
X_{\ell, j}.
\]
Our aim will be to show that $\E[X] = o(1)$.

In order for a set $S \subseteq V_H$ to contribute to~$X_{\ell,j}$
its size, $s$, must be at most $(\ell-j)/2+1$ (otherwise, the
$(\ell-j)/2$ internal edges won't be able to connect~$S$).
The number of size-$s$ sets $S$ is at most
$\binom{n}{s}$ which is at most $ \binom{n}{(\ell - j)/2 + 1}$,
as long as $n$ is sufficiently large. The number of possible 
values of~$s$ is at most $(\ell-j)/2 \leq \ell$.
Thus, the number of possibilities for~$S$ is
at most   $\ell \cdot \binom{n}{(\ell - j)/2 + 1}$ 
  
Given $S$, there are at most $\binom{\ell}{j}$ ways to choose the half-edges that will make up $E_H(S, S^c)$ once $H$ is drawn. The remaining total degree of $S$ after choosing these edges is $\ell - j$. By Lemma~\ref{lem:frprob}, the probability that every other edge of of $E_H(S, V_H)$ is an edge of $E_H(S, S)$ is at most $\binom{m'}{(\ell-j)/2}^{-1}$, where $m' = m - j$. As noted in Case~1, it is shown using Stirling's approximation (see equation $7$) in~\cite{fountoulakis2007evolution} 
that there exists a constant $C$ such that
for all $n$, $\ell$, $j$ and~$m'$,
\[
\binom{n}{(\ell - j)/2 + 1} \binom{m'}{(\ell - j)/2}^{-1} \leq C m' \left( \frac{n}{m'} \right)^{\frac{\ell - j}{2} + 1}.
\]
Since 
$2m \geq 3n$ by  Item~\ref{item:i} of Definition~\ref{def:degseq} 
and $2j \leq 2\ell/100 \leq n/100$,
$2m' = 2m - 2j \geq 3n - n/100$ so $n/m' < 7/10$. Thus,
\[
\binom{n}{(\ell - j)/2 + 1} \binom{m'}{(\ell - j)/2}^{-1} \leq C m' \left( \frac{7}{10} \right)^{\frac{\ell - j}{2} + 1} \leq C m' \left( \frac{7}{10} \right)^{\frac{\ell - \ell/100}{2}} \leq C m' \left( \frac{17}{20} \right)^{\ell}.
\]
Furthermore, we have that
\[
\binom{\ell}{j} \leq \binom{\ell}{\ell/100} \leq (100\emm)^{\ell/100}.
\]
Putting things together, it follows that
\begin{align*}
\E[X] &\leq \sum_{\ell = \ceil{(\log n)^2}}^{\floor{n/2}} \sum_{j \in J_\ell} \ell \cdot \binom{\ell}{j} \binom{n}{(\ell - j)/2 + 1} \binom{m'}{(\ell - j)/2}^{-1} \\
&\leq C \sum_{\ell = \ceil{(\log n)^2}}^{\floor{n/2}} 
\ell^2 m' \left( \frac{17}{20} \right)^{\ell}(100\emm)^{\ell/100}.
\end{align*}
Now $m' \leq m$ and this is $O(n)$ by Item~\ref{item:ii} of Definition~\ref{def:degseq}, so
$$
\E[X]
= O(n) \sum_{\ell = \ceil{(\log n)^2}}^{\floor{n/2}} \ell^2 \cdot \left( \frac{9}{10} \right)^\ell = o(1).
$$  
Applying Markov's inequality, we have that $\Pr(X > 0) = \Pr(X \geq 1) = o(1)$, and the result follows.

\textbf{Case 3.} Finally, we deal with vertex sets $S \subseteq V_H$ with $(\log n)^2 \leq |S| \leq n/2$ and $\deg_H(S) > \max\{ 100 d |S|, n/2\}$. 
Let $C = 10^4 d$.
By the Cauchy-Schwarz inequality, we have that $|S|\sum_{i\in S} x_i^2\geq (\deg_H(S))^2\geq 10^4d^2|S|^2$, so using Item~\ref{item:ii} of Definition~\ref{def:degseq} 
which ensures that $ \sum_{i\in S}x_i^2 \leq d n$,
we find that $|S|\leq n/C$. 

Let $f=(\emm C)^{1/C}$.
The number of sets $S$ satisfying $|S|\leq n/C$
is at most $n\binom{n}{n/C} \leq n f^n$
since there are at most $n$ possibilities for $|S|$ to consider, and  
for each of them $\binom{n}{|S|} \leq \binom{n}{n/C}$.

Fix any set $S \subseteq V_H$ with $|S|\leq n/C$
and consider the random construction of~$H$, starting from half-edges in~$S^c$ (and choosing their mates in the pairing).
Let 
$$j = \left\lfloor\frac{\deg_H(S^c)}{2}\right\rfloor \geq \left\lfloor \frac{3 |S^c|}{2}\right\rfloor \geq 
\left\lfloor \frac{3n \big(1-\frac{1}{C}\big)}{2}\right\rfloor \geq \frac{3n \big(1-\frac{2}{C}\big)}{2},$$
where the first inequality uses the fact that each $x_i$ is at least~$3$ (from Item~\ref{item:i} of
Definition~\ref{def:degseq}) and the
final inequality uses the fact that $n$ is sufficiently large.

Note that the process initiates a pairing from at least $j$ half-edges in~$S^c$.
For each $i \in [j]$, let $Y_i$ be the indicator random variable for the event that the $i$'th
half-edge from which pairing is initiated connects to an endpoint in~$S$ 
(conditioned on the pairings of the first $i-1$ half-edges initiated from~$S^c$).

Let $\epsilon = 3(1-2/C)/(8\sqrt{d})\leq 1/2$. 
From Item~\ref{item:ii} in Definition~\ref{def:degseq} and
the Cauchy-Schwarz inequality, 
${(\sum x_i)}^2 \leq n \sum x_i^2 \leq n d$ so
$\sum_{i=1}^n x_i \leq \sqrt{d} n$.
For any $t\in [j]$ satisfying $\sum_{i=1}^{t-1}Y_t < \epsilon n/2$
we have
$$\Pr(Y_t = 1)\geq \frac{\deg_H(S) - \epsilon n/2}{\sqrt{d}n} > \frac{1-\epsilon}{2\sqrt{d}}\geq \frac{1}{4\sqrt{d}}.$$
Now let $X_1,\ldots,X_j$ be i.i.d.{} Bernoulli random variables which are~$1$
with probability~$ 1/(4\sqrt{d})$.
We can couple the evolution of these variables so that, for any $t\in[j]$ satisfying
$\sum_{i=1}^{t-1}Y_i < \epsilon n/2$, we have
$\sum_{i=1}^t Y_i \geq \sum_{i=1}^t X_i$. 
We conclude that $\Pr(\sum_{i=1}^j Y_i < \epsilon n/2) \leq \Pr(\sum_{i=1}^j X_i < \epsilon n/2)$.
 
To conclude we will show that
$n f^n \Pr(\sum_{i=1}^j X_i < \epsilon n/2)  = o(1)$, implying that we can take 
$\alpha'' = \epsilon/(2\sqrt{d})$ since $\epsilon n/2 = \alpha'' \sqrt{d} n \geq \alpha'' \deg_H(S)$.

Let $X = \sum_{i=1}^j X_i$ and $\delta=1/2$.
Note that $\E[X] = j/(4\sqrt{d})$
and that   $$\frac{(1-\delta) j}{4\sqrt{d}} \geq 
\frac{(1-\delta) 3n \big(1-\frac{2}{C}\big)}{8\sqrt{d}} 
=  \frac{\epsilon n}{2}.$$
By a Chernoff bound, 
$\Pr(X \leq \epsilon n/2) \leq 
\Pr(X \leq (1-\delta) j/(4\sqrt{d})) \leq  \exp(-j\delta^2/(8\sqrt{d}))$.

To conclude that $n f^n \exp(-j\delta^2/(8\sqrt{d})) = o(1)$
we observe 
that $f < \exp(3(1-2/C)\delta^2/(16\sqrt{d}))$. 
So with $\alpha''=\epsilon/(2\sqrt{d})$, we conclude Case~3.

To prove the lemma, combine the three cases by taking
$\alpha = \min\{\alpha'/(100 d),1/100,\alpha''\}$.

\end{proof}
\end{lemma}

We can now prove the following result, which establishes the desired expansion properties of the multigraphs generated by the configuration model.

\begin{lemma}\label{lem:lem}
Let $d$ be a positive real number. There is a positive real number $ \alpha$ (depending on~$d$) such that  the following is true with high probability when $H = (V_H, E_H)$ is drawn from $\cm(n, \vec{x})$ uniformly over all degree sequences  $\vec{x}   \in \calD_{n, d}$. For all connected vertex sets $S \subseteq V_H$ with $|S| \leq n/2$, we have that $e_H(S, S^c) \geq \alpha \deg_H(S)$.
\end{lemma}
\begin{proof} 
We consider three cases.
\begin{description}
 \item {\bf Case 1.}
Consider all connected subsets $S \subseteq V_H$   with $(\log n)^2 \leq |S| \leq n/2$. By Lemma~\ref{lem:bigexpansion} there is a positive real number $\alpha'$ such that, whp, every such subset~$S$ has $e_H(S, S^c) \geq \alpha' \deg_H(S)$. 

\item {\bf Case 2.}
Consider all connected subsets
$S \subseteq V_H$   with $|S| \leq (\log n)^2$ and $\deg_H(S) \geq 36$.

\begin{itemize}
    \item   Consider first 
    those subsets $S$ with $|S| \leq \frac{1}{6} \deg_H(S)$. We have that
\[
e_H(S, S^c) = \deg_H(S) - 2(t_{H[S]} + |S| - 1) \geq \frac{2}{3} \deg_H(S) - 2|S| \geq \frac{1}{3} \deg_H(S),
\]
by Lemma~\ref{lem:excess} and our assumption on the size of $S$. 
\item Now consider those subsets $S$ with  $|S| > \frac{1}{6} \deg_H(S)$, then by Lemma~\ref{lem:smallexpansion}, we have that
$
e_H(S, S^c) \geq |S|/4 \geq \deg_H(S)/24$. \end{itemize}
\item {\bf Case 3.}
Finally, consider connected subsets $S\subseteq V_H$ with   
 $|S| \leq (\log n)^2$ and
$\deg_H(S) < 36$. 

By Lemma~\ref{lem:smallexpansion}, we have that
$
e_H(S, S^c) \geq |S|/4 \geq 1/4 = 36/144 > \deg_H(S)/144$.
\end{description}
The result follows from the three cases by taking $\alpha = \min \{ 1/144, \alpha' \} = \alpha'$. 
\end{proof} 

Using the definition of $\calG_\alpha$ and
Lemma~\ref{lem:cmtog}, we have the following corollary of  Lemma~\ref{lem:lem}.

\begin{corollary}
\label{cor:expander}
Let $d$ be a real number. There  is a positive real number 
$\alpha$ (depending on~$d$) such that the following holds. With high probability, when $G\sim \calG(n, \vec{x})$ for some  $\vec{x} \in \calD_{n, d}$, it holds that $G\in \calG_\alpha$. 
\end{corollary}

Combining Corollary~\ref{cor:expander} with Theorem~\ref{thm:expandersampling} implies our main theorem.

\randgrsampling*
\begin{proof}
Let $d$ be a real number and let $q\geq 2$ be an integer.
Let $\alpha$ be the positive real number from Corollary~\ref{cor:expander}.
Let $\beta_0 = \tfrac{3}{\alpha} \log(8 \emm^3(q-1))$.

Consider 
$\vec{x} \in \calD_{n, d}$
and let~$G$ be drawn  from $\calG(n, \vec{x})$.    
By Corollary~\ref{cor:expander}, $G \in \calG_\alpha$ whp. The result then follows by using the algorithms  from Theorem~\ref{thm:expandersampling}.
\end{proof}

\begin{remark}\label{rem:remarkdetails}
The bounds on~$\beta$ in Remark~\ref{rem:ourrem} follow from the choice of~$\beta_0$ in the proof
of Theorem~\ref{thm:randgrsampling} and from the fact that $\alpha = \Omega(\tfrac{1}{d\log d})$
which follows from the proofs of Lemmas~\ref{lem:bigexpansion} and~\ref{lem:lem}.
The running time bounds in Remark~\ref{rem:ourrem} come from those
in Theorem~\ref{thm:expandersampling} using the fact that $|E_G| = O(n)$ which follows from Item~\ref{item:ii} of Definition~\ref{def:degseq}.
\end{remark}

\section{Details of the algorithms for the edge-based polymer model}
\label{sec:details}

We start with a lemma giving an upper bound on the number of
vertex subsets of a given total degree.
This lemma is used throughout this section.

\npolymers*

\begin{proof}
Let $N(G, v, \ell)$ be the set of subtrees $T = (V_T, E_T)$ of $G$ such that $v \in V_T$, $\deg_G(V_T) = \ell$. We will show that $|N(G, v, \ell)| \leq (2\emm)^{2 \ell - 1}$, which gives us the desired result for the following reason. Let $S \subseteq V_G$ be a connected vertex subset such that $v \in S$ and $\deg_G(S) = \ell$. Since $S$ is connected, it has at least one spanning tree $T = (V_T = S, E_T)$ such that $v \in V_T$ and $\deg_G(V_T) = \ell$. Since $S$ is the unique connected vertex subset that $T$ spans, this gives us an injective map from the set of all connected vertex subsets containing $v$ with total degree $\ell$, to $N(G, v, \ell)$. 

We now give an injective map from $N(G, v, \ell)$ to $T^*(2\ell, 3)$ -- the set of subtrees of size $2\ell$ that contain the root, of the infinite rooted $3$-regular tree. By a result of Bollob{\'a}s~\cite[p. 129]{bollobas2006art}, we know that $|T^*(2 \ell, 3)|$ is at most $(2\emm)^{2 \ell - 1}$. Let $T = (V_T, E_T)$ be a subtree from $N(G, v, \ell)$. We will map $T$ to a rooted subtree $T' = (V_{T'}, E_{T'})$ from $T^*(2\ell, 3)$. For each vertex of $V_G$, fix an ordering of its neighbours. In the infinite rooted $3$-regular tree, label the edges incident to the root with $\{ 1, 2, 3 \}$, and for each other vertex label the edges connecting it to its two children with $\{ 1, 2 \}$. As we construct $T'$, we will label its edges so that it is clear which subtree from $T^*(2 \ell, 3)$ we are constructing, we will also label some of its vertices. We construct $T'$ as follows (see Figure~\ref{fig:trees} for an example of the following construction).
\begin{enumerate}
\item Add the root to $V_{T'}$ and label it $v$.
\item While there is a labelled vertex of $T'$ (call its label $u$) such that $u$ has a child $w$ in $T$ but no vertex of $T'$ is labelled $w$, then we do the following. First, we create a path $P$ of length $\deg_G(u)$ where each edge is labelled $1$. We then connect the vertex of $T'$ labelled $u$ to $P$ via an edge labelled $1$. Finally, for $1 \leq i \leq \deg_G(u)$, we connect a vertex labelled $w$  to the $i^{\text{th}}$ vertex of $P$ via an edge labelled $2$, if $w$ is the $i^{\text{th}}$ neighbour of $u$ in $G$ and $w$ is a child of $u$ in $T$.
\end{enumerate}

\begin{figure}[ht]
\begin{subfigure}[b]{0.32\textwidth}
\centering
\raisebox{10mm}{
\begin{tikzpicture}[sibling distance=1em, every node/.style = {shape = circle, fill = black, minimum size = 7pt, inner sep = 0pt}]]
\node[label={\small $u$}] {}
child { node[label=below:{\scriptsize $u_1$}] {} }
child { node[label=below:{\scriptsize $u_2$}] {} }
child { node[label=below:{\scriptsize $u_3$}] {} }
child { node[label=below:{\scriptsize $u_4$}] {} }
child { node[label=below:{\scriptsize $u_5$}] {} }
child { node[label=below:{\scriptsize $u_6$}] {} };
\end{tikzpicture}
}
\caption{Neighbourhood of $u$ in $G$}
\end{subfigure}
\begin{subfigure}[b]{0.32\textwidth}
\centering
\raisebox{10mm}{
\begin{tikzpicture}[sibling distance=1em, every node/.style = {shape = circle, fill = black, minimum size = 7pt, inner sep = 0pt}]]
\node[label={\small $u$}] {}
child { node[label=below:{\scriptsize $u_1$}] {} }
child { node[label=below:{\scriptsize $u_3$}] {} }
child { node[label=below:{\scriptsize $u_5$}] {} }
child { node[label=below:{\scriptsize $u_6$}] {} };
\end{tikzpicture}
}
\caption{Neighbourhood of $u$ in $T$}
\end{subfigure}
\begin{subfigure}[b]{0.32\textwidth}
\centering
\raisebox{3mm}{
\begin{tikzpicture}
\node[Vtx, label=above :{$u$}] (u) at (0,0){};

\node[Vtx, label=above :{}] (u1) at (0.5,-0.5){};
\node[Vtx, label=left :{\scriptsize $u_1$}] (u1h) at (0.1,-1.2){};

\node[Vtx, label=above :{}] (u2) at (1,-1){};

\node[Vtx, label=above :{}] (u3) at (1.5,-1.5){};
\node[Vtx, label=left :{\scriptsize $u_3$}] (u3h) at (1.1,-2.2){};

\node[Vtx, label=above :{}] (u4) at (2,-2){};

\node[Vtx, label=above :{}] (u5) at (2.5,-2.5){};
\node[Vtx, label=left :{\scriptsize $u_5$}] (u5h) at (2.1,-3.2){};

\node[Vtx, label=above :{}] (u6) at (3,-3){};
\node[Vtx, label=left :{\scriptsize $u_6$}] (u6h) at (2.6,-3.7){};
    
\draw (u) --node[above right=-0.5ex] {\scriptsize $1$} (u1) --node[above right=-0.5ex] {\scriptsize $1$} (u2) --node[above right=-0.5ex] {\scriptsize $1$} (u3) --node[above right=-0.5ex] {\scriptsize $1$} (u4) --node[above right=-0.5ex] {\scriptsize $1$} (u5) --node[above right=-0.5ex] {\scriptsize $1$} (u6);
\draw (u1) --node[above left=-0.5ex] {\scriptsize $2$} (u1h);
\draw (u3) --node[above left=-0.5ex] {\scriptsize $2$} (u3h);
\draw (u5) --node[above left=-0.5ex] {\scriptsize $2$} (u5h);
\draw (u6) --node[above left=-0.5ex] {\scriptsize $2$} (u6h);
\end{tikzpicture}
}
\caption{Neighbourhood of $u$ in $T'$}
\end{subfigure}
\caption{}
\label{fig:trees}
\end{figure}

Each $T \in N(G, v, \ell)$ maps to a different $T' \in T^*(2 \ell, 3)$. When constructing $T'$, we used edge labels from $\{ 1, 2, 3 \}$, therefore the maximum degree of $T'$ is $3$. For each $v \in V_T$, we added at most $2 \deg_G(v)$ vertices to $T'$, therefore the size of $T'$ is at most $2 \deg_G(V_T) = 2 \ell$.
\end{proof}  

The following lemma shows that if a family of polymer models satisfies the polymer sampling condition then it also satisfies the polymer mixing condition. This will be convenient when applying our algorithms, as it is in general easier to show that the polymer sampling condition holds for a given family of polymer models, than to show directly that the polymer mixing condition does.
However, the polymer mixing condition is used to bound the
mixing time of the polymer dynamics.

For a vertex subset $S \subseteq V_G $, we let $\partial_G S$ denote the vertices of $S^c$ that are joined to $S$ by an edge.

\begin{lemma}
\label{lem:samplingtomixing}
Let $q\geq 2$ be an  integer, $\calG$ be a class of graphs, and $\calF_\calG = \{ (\calC_G, w_G) \mid G \in \calG \}$ be a family of $q$-spin polymer models. If $\calF_\calG$ satisfies the polymer sampling condition (Definition~\ref{def:psample}) then $\calF_\calG$ satisfies the polymer mixing condition (Definition~\ref{def:pmix}) with constant $\theta = \frac{1}{e}$.
\end{lemma}
\begin{proof}
Let $G \in \calG$ and let $\gamma \in \calC_G$ be an arbitrary polymer. We have that
\[
\sum_{\gamma' \nsim \gamma} |E_{\gamma'}| \cdot w_G(\gamma') \leq \sum_{\{ u, v\} \in E_\gamma} \left( \sum_{\gamma' : u \in V_{\gamma'}} |E_{\gamma'}| \cdot w_G(\gamma') + \sum_{\gamma' : v \in V_{\gamma'}} |E_{\gamma'}| \cdot w_G(\gamma') \right),
\]
therefore the result will follow if we are able to show for all $v \in V_\gamma \cup \partial_G V_\gamma$ that
\[
\sum_{\gamma' : v \in V_{\gamma'}} |E_{\gamma'}| \cdot w_G(\gamma') \leq \frac{1}{2\emm}.
\]
We can re-write the left-hand side of the above as follows
\[
\sum_{\gamma' : v \in V_{\gamma'}} |E_{\gamma'}| \cdot w_G(\gamma') \leq \sum_{\ell \geq 1} \sum_{\substack{\gamma' :~v \in V_{\gamma'}, \\ \deg_G(V_{\gamma'}) = \ell}} |E_{\gamma'}| \cdot w_G(\gamma').
\]
By Lemma~\ref{lem:npolymers}, we know that there are at most $(2\emm)^{2\ell - 1}$ connected vertex subsets $S$ such that $v \in S$, and $\deg_G(S) = \ell$. For each such $S$ there are $(q-1)^{|S|} \leq (q-1)^\ell$ ways to assign spins to its vertices, therefore there are at most $(2\emm)^{2\ell - 1}(q-1)^\ell$ polymers $\gamma'$ such that $v \in V_{\gamma'}$ and $\deg_G(V_{\gamma'}) = \ell$. Furthermore, since $\calF_\calG$ satisfies the polymer sampling condition, we know that each of these polymers satisfies $w_G(\gamma') \leq \emm^{-\tau \ell}$ where $\tau \geq 3\log(8 \emm^3 (q-1))$. Combining these facts with the observation that $|E_{\gamma'}| \leq 
\deg_G(V_{\gamma'}) \leq
\emm^{\deg_G(V_{\gamma'})}$ for all $\gamma' \in \calC_G$, we obtain that
\[
\sum_{\gamma' : v \in V_{\gamma'}} |E_{\gamma'}| \cdot w_G(\gamma') 
\leq \sum_{\ell \geq 1} (2\emm)^{2\ell-1} (q-1)^\ell \emm^\ell \emm^{-\tau \ell} \leq \frac{1}{2\emm} \sum_{\ell \geq 1} \left( \frac{4\emm^3 (q-1)}{\emm^{\tau}} \right)^\ell.
\]
Finally, since $\tau \geq \log(8 \emm^3 (q-1))$, the result follows. 
\end{proof} 

\subsection{Polymer dynamics}

We will define a Markov chain (the polymer dynamics) to sample from $\mu_G$, that mixes rapidly if the polymer mixing condition is met. This Markov chain is essentially the same as the polymer dynamics of~\cite{chen2019fast}, except that we sample edges instead of vertices. For $e \in E_G$ let $\calA(e)$ be the set of polymers $\gamma$ such that $e \in E_\gamma$. Let $\nu_e$ be the probability distribution on $\calA(e) \cup \emptyset$ defined by
\[
\nu_e(\gamma) = w_G(\gamma) \quad \text{ and} \qquad \nu_e(\emptyset) = 1 - \sum_{\gamma \in \calA(e)} \nu_e(\gamma).
\]
If the polymer mixing condition is satisfied, then $\nu_e$ gives a well-defined probability distribution. To see this, let $\gamma \in \calA(e)$ be such that $|E_{\gamma}|$ is least. Applying the polymer mixing condition, we obtain that
\[
\sum_{\substack{\gamma' \nsim \gamma \\ \gamma' \in \calA(e)}} |E_{\gamma'}| \cdot w_{G}(\gamma') \leq \sum_{\gamma' \nsim \gamma} |E_{\gamma'}| \cdot w_G(\gamma') < |E_\gamma|.
\]
Thus,
\[
\sum_{\gamma '\in \calA(e)} \nu_e(\gamma') = \sum_{\substack{\gamma' \nsim \gamma \\ \gamma' \in \calA(e)}} w_G(\gamma') \leq \sum_{\substack{\gamma' \nsim \gamma \\ \gamma' \in \calA(e)}} \frac{|E_{\gamma'}|}{|E_\gamma|} \cdot w_G(\gamma') \leq 1.
\]

Let $t \geq 0$ be an integer. The polymer dynamics Markov chain transitions from $\Gamma_t \in \Omega_G$ to $\Gamma_{t+1} \in \Omega_G$ according to the following rules.

\begin{enumerate}
\item Choose $e \in E_G$ uniformly at random. Let $\gamma_e$ be the (unique) polymer in $\Gamma_t \cap \calA(e)$, if it exists; otherwise, let $\gamma_e = \emptyset$.
\item Mutually exclusively
\begin{enumerate}
\item with probability $1/2$, let $\Gamma_{t+1} = \Gamma_t \setminus \{ \gamma_e \}$, and
\item with probability $1/2$, sample a polymer $\gamma$ from $\nu_e$, and let $\Gamma_{t+1} = \Gamma_t \cup \{ \gamma \}$ if this is compatible; otherwise, let $\Gamma_{t+1} = \Gamma_t$.
\end{enumerate}
\end{enumerate}

The following lemma shows that the unique stationary distribution of the polymer dynamics is the Gibbs distribution of the polymer model.

\begin{lemma}
\label{lem:stdist}
Let $G = (V_G, E_G)$ be a graph and let $(\calC_G, w_G)$ be a polymer model. The unique stationary distribution of the polymer dynamics is $\mu_G$.
\begin{proof}
One update of the polymer dynamics changes the polymer configuration by at most one polymer. Let $\Gamma' = \Gamma \cup \{ \gamma \}$, then
\[
\frac{\mu_G(\Gamma')}{\mu_G(\Gamma)} = w_G(\gamma) = \frac{w_G(\gamma)|E_\gamma|/2|E_G|}{|E_\gamma|/2|E_G|} = \frac{\Pr(\Gamma \rightarrow \Gamma')}{\Pr(\Gamma' \rightarrow \Gamma)}.
\]
It therefore follows by detailed balance that $\mu_G$ is a stationary distribution of the polymer dynamics.

Finally, the polymer dynamics are irreducible since we can move from any $\Gamma \in \Omega_G$ to any $\Gamma' \in \Omega_G$ by adding and removing polymers (for example, via the empty set). The polymer dynamics are also aperiodic since there are self-loops. 
\end{proof}
\end{lemma}

The following lemma shows that if the polymer mixing condition is satisfied, then the polymer dynamics Markov chain mixes rapidly.\footnote{Recall that for an ergodic Markov chain  with a finite state space $\Omega$ and transition matrix~$P$,
the mixing time to its stationary distribution $\mu$ is defined as $\max_{\omega\in \Omega}\min \{ t > 0 : \| P^t(\omega,\cdot) - \mu  \|_{TV} \leq \epsilon \}$, where $\| \pi - \pi' \|_{TV}$ is the total variation distance between the probability distributions $\pi$ and $\pi'$.}

\begin{lemma}
\label{lem:mixing}
Let $\calG$ be a class of graphs and let $\{ (\calC_G, w_G) \mid G \in \calG \}$ be a family of polymer models that satisfies the polymer mixing condition (see Definition~\ref{def:pmix}). For all $G \in \calG$ and all $\epsilon > 0$, the mixing time of the polymer dynamics is $\tmix(\epsilon) = O(|E_G| \log (|E_G|/\epsilon))$.
\begin{proof}
Let $G \in \calG$. We proceed by path coupling. Let $D(\cdot, \cdot)$ be a metric on $\Omega_G$ which we define by setting $D(\Gamma, \Gamma') = 1$ for polymer configurations $\Gamma, \Gamma' \in \Omega_G$ such that $\Gamma' = \Gamma \cup \{ \gamma \}$ for some polymer $\gamma \in \calC_G$. This can be extended to a shortest path metric on $\Omega_G$ (the symmetric difference).

Let $(X_t, Y_t)$ be a coupling for the polymer dynamics where both chains make the same random choices. Suppose that $D(X_t, Y_t) = 1$ and that $X_t = Y_t \cup \{ \gamma \}$ for some $\gamma \in \calC_G$. We obtain $(X_{t+1}, Y_{t+1})$ from $(X_t, Y_t)$ by performing an update according to the polymer dynamics, and consider the various possibilities for $D(X_{t+1}, Y_{t+1})$. With probability $|E_\gamma|/2|E_G|$ we choose $\gamma$ and remove it from both $X_t$ and $Y_t$ -- in this case $D(X_{t+1}, Y_{t+1}) = 0$. Alternatively, if we add a polymer $\gamma'$ such that $\gamma' \nsim \gamma$, then the update will be rejected in one chain and performed in the other and $D(X_{t+1}, Y_{t+1}) \leq 2$ -- this occurs with probability $|E_{\gamma'}|/2|E_G| \cdot w_G(\gamma')$. Combining these two cases, we obtain that
\[
\E[D(X_{t+1}, Y_{t+1}) \mid (X_t, Y_t)] \leq 1 + \frac{1}{2|E_G|} \left( - |E_\gamma| + \sum_{\gamma' \nsim \gamma} |E_{\gamma'}| \cdot w_G(\gamma') \right),
\]
and therefore by the polymer mixing condition that
\[
\E[D(X_{t+1}, Y_{t+1}) \mid (X_t, Y_t)] \leq 1 - \frac{|E_\gamma| (1 - \theta)}{2 |E_G|} \leq 1 - \frac{(1 - \theta)}{2 |E_G|},
\]
for some constant $\theta \in (0, 1)$. Let $W$ be the diameter of $\Omega_G$ with respect to the metric $D(\cdot, \cdot)$ and note that $W \leq 2 n$. By the path coupling lemma~\cite[Section 6]{dyer1999random}, it follows that the mixing time is at most $\log(W/\epsilon)2|E_G|/(1 - \theta) = O(|E_G| \log(|E_G|/\epsilon))$.
\end{proof}
\end{lemma}

\subsection{Single polymer sampler}

We now show how to use the polymer dynamics as part of an efficient algorithm for computing an approximate sample from $\mu_G$. The input to this algorithm will be a graph and an accuracy parameter, and we will assume that we are given the graph both as an adjacency matrix and an array of linked adjacency lists, where each 
adjacency list contains 
all of the neighbours of a vertex. These data structures will allow us to decide whether two given vertices are adjacent in constant time (in the unit cost model), 
and will also allow us to access the next entry in  any adjacency list   in constant time. Other choices of data structures will increase the running time of our algorithms by at most a polynomial factor.

Fix $e \in E_G$. We now show how to sample a polymer from $\nu_e$ in expected constant time. Let $r = \tau - \log(12\emm^2(q-1)) > 0$, where $\tau$ is as in the polymer sampling condition (Definition~\ref{def:psample}). For $\ell \geq 0$ let $\calA_\ell(e) = \{ \gamma \in \calA(e) : \deg_G(V_\gamma) \leq \ell \}$. We use the following algorithm to to sample a polymer from $\nu_e$.

\begin{center}
\begin{algorithm}[h]
\caption{Sampling a polymer from $\nu_e$.}
\label{alg:pmersampler}
\begin{enumerate}[leftmargin=\parindent, labelsep=5pt]
\item Sample $\bell$ from the following distribution: $\Pr(\bell = \ell) = (1 - \emm^{-r}) \emm^{-r\ell}$ for all integers $\ell \geq 1$. This is so that $\Pr(\bell \geq \ell) = \emm^{-r \ell}$.
\item Enumerate $\calA_{\bell}(e)$ and compute $w_G(\gamma)$ for each $\gamma \in \calA_{\bell}(e)$.
\item Mutually exclusively, output each $\gamma \in \calA_{\bell}(e)$ with probability $w_G(\gamma) \cdot \emm^{r \deg_G(\gamma)}$. With all remaining probability, output $\emptyset$.
\end{enumerate}
\end{algorithm} 
\end{center}

In order to perform the second step of the above algorithm, we must be able to enumerate all connected vertex subsets $S \subseteq V_G$ with $\deg_G(S) \leq \bell$, that contain an endpoint of $e$. In order to do this efficiently, we require the following result which is an adaptation of Lemma 3.4 of~\cite{patel2017deterministic}.

\begin{lemma}
\label{lem:enumsgraphs}
Let $G = (V_G, E_G)$ be a graph, let $v \in V_G$, and let $\ell \geq 1$. There is an algorithm which enumerates all connected vertex subsets $S \subseteq V_G$ such that $v \in S$ and $\deg_G(S) \leq \ell$. The running time of this algorithm is $O(\ell^7 (2\emm)^{4 \ell})$.
\begin{proof}
Let $C(G, v, \ell)$ be the set of all connected vertex subsets $S \subseteq V_G$ such that $v \in S$ and $\deg_G(S) \leq \ell$. For $\ell < \deg_G(v)$, we simply output $\emptyset$. For $\ell \geq \deg_G(v)$, we construct $C(G, v, \ell)$ recursively. For the base case, observe that $C(G, v, \deg_G(v))$ is $\{ v \}$. For $\ell' \geq \deg_G(v)$, in order to construct $C(G, v, \ell'+1)$ given $C(G, v, \ell')$, we first construct the multiset
\[
C''(G, v, \ell' + 1) = \{ S \cup \{ u \} : S \in C(G, v, \ell'),~u \in \partial_G S,~\deg_G(S \cup \{u\}) = \ell' + 1 \},
\]
then remove the repeat elements from it to obtain $C'(G, v, \ell' + 1)$, and finally set $C(G, v, \ell' + 1) = C'(G, v, \ell' + 1) \cup C(G, v, \ell')$.

To construct $C''(G, v, \ell' + 1)$, we consider each element $S \in C(G, v, \ell')$ and each vertex $u \in \partial_G S$. Iterating through $C(G, v, \ell')$ requires $O(\ell' (2\emm)^{2\ell' - 1})$ time, by Lemma~\ref{lem:npolymers}. For each 
$S \in C(G, v, \ell')$, we can iterate through $\partial_G S$ in $O(\ell')$ time, given that we have constant-time access to the elements of the adjacency list of each vertex in $S$. The total time required is therefore $O(\ell'^2 (2\emm)^{2\ell' - 1})$. We then remove the repeat elements from $C''(G, v, \ell' + 1)$ by pairwise comparison, which takes $O((\ell' + 1)^2 \cdot |C''(G, v, \ell' + 1)|^2) = O(\ell'^6 (2\emm)^{4 \ell'})$ time.

To construct $C(G, v, \ell)$, we begin with the base case and then perform the above $\ell - \deg_G(v)$ times. The total running time is therefore $O(\ell^7 (2\emm)^{4 \ell})$. We now show that the algorithm returns every element of $C(G, v, \ell)$. For the base case, it is clearly true. Furthermore, for all $\ell' > \deg_G(v)$ and all $S \in C(G, v, \ell')$, we know that $S$ consists of $T \in C(G, v, \ell'')$ and $u \in \partial_G T$, for some $\ell'' < \ell'$.
\end{proof}
\end{lemma}

We can now show the correctness and expected constant running time of Algorithm~\ref{alg:pmersampler}. 

\begin{lemma}
\label{lem:spsampler}
Let $q \geq 2$ be an integer,  $\calG$ be a class of graphs, and  $\{ (\calC_G, w_G) \mid G \in \calG \}$ be a family of computationally feasible $q$-spin polymer models that satisfies the polymer sampling condition (see Definition~\ref{def:psample}). For all $G \in \calG$ and all $e \in E_G$, Algorithm~\ref{alg:pmersampler} samples a polymer from $\nu_e$ in expected constant time.
\begin{proof}
We begin by showing that step 3 of the algorithm is well-defined, by showing that
$\sum_{\gamma \in \calA(e)} w_G(\gamma) \cdot \emm^{r \deg_G(V_\gamma)} \leq 1$ for all $e \in E_G$. For all $v \in V_G$, we have that
\[
\sum_{\gamma : u \in V_{\gamma}} w_G(\gamma) \cdot \emm^{r \deg_G(V_\gamma)} \leq \sum_{\ell \geq 1} \sum_{\substack{\gamma : u \in V_{\gamma}, \\ \deg_G(V_\gamma) = \ell}} w_G(\gamma) \cdot \emm^{r \ell}\leq \sum_{\ell \geq 1} (4\emm^2)^\ell (q-1)^\ell \emm^{- \tau \ell} \emm^{r \ell} \leq \frac{1}{2},
\]
where the second-to-last inequality follows from Lemma~\ref{lem:npolymers} and the final inequality follows from the fact that $4\emm^2(q-1)\emm^{- \tau + r} \leq 1/3$.

We now show that the expected running time of the algorithm is constant. For $\ell \geq 1$, the time taken to enumerate $\calA_\ell(e)$ is $O((q-1)^\ell \ell^7 (2\emm)^{4 \ell})$, by Lemma~\ref{lem:enumsgraphs}. The time taken to then iterate through $\calA_\ell(e)$ and compute $w_G(\gamma)$ for each $\gamma \in \calA_\ell(e)$ is $O((q-1)^\ell \ell (2\emm)^{2 \ell - 1} \emm^\ell)$, by Lemma~\ref{lem:npolymers} and the fact that the family of polymer models is computationally feasible. The expected running time of Algorithm~\ref{alg:pmersampler} is therefore 
\begin{align*}
&= O \left( \sum_{\ell \geq 1} \Pr(\bell = \ell) \left( (q-1)^\ell \ell^7 (2\emm)^{4 \ell} +  (q-1)^\ell \ell (2\emm)^{2 \ell - 1} \emm^\ell \right) \right) \\
&= O \left( \sum_{\ell \geq 1} \ell^7 \left( \frac{16\emm^4(q-1)}{\emm^r} \right)^\ell \right).
\end{align*}
Since $r = \tau - \log(12\emm^2(q-1)) \geq 3\log(8\emm^3(q-1)) - \log(12\emm^2(q-1)) \geq 2 \log(8 \emm^3 (q-1))$, it follows that $16\emm^4(q-1)/\emm^r < 1/2$, and therefore that the expected running time of Algorithm~\ref{alg:pmersampler} is $O(1)$. 

Finally, let $\gamma \in \calA(e)$ be a polymer. In order for the algorithm to sample $\gamma$, it must first  sample $\bell \geq \deg_G(V_\gamma)$ in step 1 of the algorithm, then conditioned on this choice, it must output $\gamma$ in step 3. This occurs with probability $\emm^{-r \deg_G(V_\gamma)} w_G(\gamma) \cdot \emm^{r \deg_G(V_\gamma)} = w_G(\gamma)$, and therefore the output distribution is $\nu_e$.
\end{proof}
\end{lemma}

We now combine the polymer dynamics with Algorithm~\ref{alg:pmersampler} to give an efficient algorithm for computing an approximate sample from the Gibbs distribution of a polymer model.

\pmersampling*
\begin{proof}
\sloppy We focus on the sampling algorithm, since it is shown  in~\cite[Section 3]{chen2019fast} how to convert this into the desired counting algorithm for the  polymer model (although the approach in~\cite{chen2019fast} was stated for bounded-degree graphs,  it carries over to graphs of unbounded degree).

By Lemma~\ref{lem:stdist}, the unique stationary distribution of the polymer dynamics is $\mu_G$, thus our sampling algorithm
is based on the polymer dynamics.
By Lemma~\ref{lem:spsampler}, there is an integer $C_1 > 1$ 
(independent of~$G$)
such that
the expected number of steps taken to execute a single update of the polymer dynamics is at most~$C_1$. 
Lemma~\ref{lem:samplingtomixing} shows that $\calF_\calG$ also satisfies the polymer mixing condition. 
This allows us to apply Lemma~\ref{lem:mixing}, 
which shows that there is an integer $C_2 > 1$ 
(also independent of~$G$)
such that 
the mixing time of the polymer dynamics satisfies
$\tmix(\epsilon/2) \leq C_2 m \log\frac{m}{\epsilon}$. 

We use the following algorithm to compute an $\epsilon$-sample from $\mu_G$. Repeat the following $\lceil \log(2/\epsilon)\rceil $ times, and if no polymer configuration is returned, return $\emptyset$. Run the polymer dynamics for $3 C_1 C_2 m \ceil{\log \tfrac{m}{\epsilon}}$ steps, starting from $\emptyset$, and if at least $C_2 m \lceil \log\tfrac{m}{\epsilon}\rceil$ updates of the polymer dynamics were executed, then return the configuration.

We claim that the distribution of the output configuration from the above algorithm is within total variation distance $\epsilon$ of $\mu_G$. This will follow once we have shown that the probability that no run of the polymer dynamics  returns a configuration, is at most $\epsilon/2$. This is because the configuration we would output if one is returned by a run of the polymer dynamics, is within total variation distance $\epsilon/2$ of $\mu_G$.

Consider one of the $\lceil \log(2/\epsilon)\rceil $ independent runs of the polymer dynamics that is made by the algorithm. Let the random variable $X$ denote the total number of steps required to execute $C_2 m \ceil{\log\frac{m}{\epsilon}}$ steps of the polymer dynamics. We have that $\E[X] \leq C_1 C_2 m \ceil{\log\tfrac{m}{\epsilon}}$, and therefore by Markov's inequality that $\Pr(X \geq 3 \E[X]) \leq 1/3<1/\emm$. Since the runs are independent, the probability that no run performs at least $C_2 m \ceil{\log\tfrac{m}{\epsilon}}$ updates of the polymer dynamics, is at most $(1/\emm)^{\log(2/\epsilon)} = \epsilon/2$.
\end{proof}

\section{Details of polymer application for the ferromagnetic Potts model}
\label{sec:pottsdetails}

It is known that if $G$ is an $\alpha$-expander
then most of the weight of $Z_{G, q, \beta}$ is contributed by colourings that colour more than half of the vertices of $G$ with a single colour. The following definitions apply to these colourings.

\begin{definition}
Let $G = (V_G, E_G)$ be a graph. Let $q$ be a positive integer, let $r \in [q]$, and let $\beta > 0$. We define $\ored$ to be set of $q$-colourings of $G$ that colour more than half of the vertices of $G$ with $r$. We also define $\zred = \sum_{\sigma \in \ored} \emm^{\beta m_G(\sigma)}$.
\end{definition}

The following result is due to Jenssen, Keevash, and Perkins~\cite{jenssen2020algorithms}. We note that although this result is applied to bounded-degree graphs in~\cite{jenssen2020algorithms}, it remains true for arbitrary $\alpha$-expanders.

\begin{lemma}\cite[Lemma 12]{jenssen2020algorithms}
\label{lem:gstates}
Let $\alpha > 0$ and let $G = (V_G, E_G)$ be an $\alpha$-expander. Let $q \geq 2$ be an integer, let $r \in [q]$ be any spin, and let $\beta > 2 \log(\emm q)/\alpha$ be a real number. We have that $q \cdot \zred$ is an $\emm^{-|V_G|}$-approximation to $Z_{G, q, \beta}$.
\end{lemma}

Consider the polymer model defined in Example~\ref{ex:Potts}. Let $\hat{\Omega}^r_{G, q}$ denote the set of all sets of mutually compatible allowed polymers, let $\hat{\mu}^r_{G, q, \beta}$ denote the Gibbs distribution of the polymer model, and let $\hat{Z}^r_{G, q, \beta}$ denote its partition function. Observe that there is a bijection between the polymer configurations of $\hat{\Omega}^{r}_{G, q}$  and the Potts configurations of $\ored$, where a colouring $\sigma\in \ored$ maps to the polymer configuration $\Gamma\in \hat{\Omega}^{r}_{G, q}$ consisting of the connected components of vertices that do not get colour $r$ under $\sigma$. Moreover, the weight of each $\Gamma \in \hat{\Omega}^{r}_{G, q}$ is closely related to the weight of the Potts configuration $\sigma \in \ored$ by
\[
\emm^{\beta |E_G|}\prod_{\gamma \in \Gamma} w_{G, \beta}(\gamma) = \emm^{\beta m_G(\sigma)}.
\]
Therefore sampling from this polymer model is equivalent to sampling from the Potts model restricted to colourings which colour more than half of the vertices with $r$.

We can now prove the following theorem, which gives an efficient algorithm for sampling from the low temperature Potts model on expanders.

\expandersampling*
\begin{proof}
Assume without loss of generality that $G$ is connected
(otherwise, consider the connected components separately).
Consider the family of polymer models defined in Example~\ref{ex:Potts}. It is computationally feasible since determining whether  $\gamma \in \calC^{r}_{G, q}$ can be done in $O(|V_\gamma|) = O(\exp(\deg_G(V_\gamma)))$ time (one just needs to check whether $|V_\gamma| < n/2$). Computing $w_{G, \beta}(\gamma) = \emm^{- \beta B_\gamma}$ can be done by examining all $O(\deg_G(V_\gamma))$ edges with endpoints in $V_\gamma$, by iterating through $V_\gamma$ and counting the neighbours of each vertex that are not in $V_\gamma$. The total running time required to do this is therefore $O(|V_\gamma|^2 \deg_H(V_\gamma)) = O(\exp(\deg_H(V_\gamma
)))$. As is shown in~\eqref{eq:wtpolymer}, since $\beta \geq \frac{3}{\alpha} \log(8 \emm^3 (q-1))$, it also satisfies the polymer sampling condition (see Definition~\ref{def:psample}) with constant $\tau = \alpha \beta$. 

The input to both algorithms is $G$ and $\epsilon$. If $\epsilon < \emm^{-n}$ then we can construct an $\epsilon$-sample from $\mu_{G, q, \beta}$ in poly$(n, \tfrac{1}{\epsilon})$ time, by brute force. If $\epsilon \geq \emm^{-n}$ then we can construct an $\epsilon$-sample from $\mu_{G, q, \beta}$ as follows. Choose a spin $r \in [q]$ uniformly at random. By Lemma~\ref{lem:pmersampling} there is an algorithm which, given $\epsilon/q$ and $G$ as input, outputs an $(\epsilon/q)$-sample from $\hat{\mu}^r_{G, q, \beta}$ (and therefore also from $\mu^r_{G, q, \beta}$) in $O\big(m \log \frac{m}{\epsilon} \log\tfrac{1}{\epsilon}\big)$ time. Since $G$ is also an $\alpha$-expander, then by Lemma~\ref{lem:gstates}, this is an $\epsilon$-sample from $\mu_{G, q, \beta}$. For the counting algorithm, if $\epsilon < \emm^{-n}$ then we can compute $Z_{G, q, \beta}$ exactly in poly$(n, \tfrac{1}{\epsilon})$ time, by brute force. If $\epsilon \geq \emm^{-n}$, then by Lemma~\ref{lem:pmersampling}, we can compute an $\epsilon$-approximation to $\zred$ in $O\big(m^2 (\log \tfrac{m}{\epsilon})^3\big)$ time. By Lemma~\ref{lem:gstates}, it follows that $q \emm^{\beta m}\cdot \zred$ is an $\epsilon$-approximation to $Z_{G, q, \beta}$.
\end{proof} 

\bibliographystyle{plain}
\bibliography{configmodel}

\end{document}